\documentclass[conference]{IEEEtran}
\makeatletter
\def\ps@headings{%
\def\@oddhead{\mbox{}\scriptsize\rightmark \hfil \thepage}%
\def\@evenhead{\scriptsize\thepage \hfil \leftmark\mbox{}}%
\def\@oddfoot{}%
\def\@evenfoot{}}
\makeatother
\usepackage{cite}
\usepackage{amsmath,amssymb,amsfonts}
\usepackage{graphicx}
\usepackage{textcomp}
\usepackage{amsthm}
\usepackage{color, colortbl}
\usepackage{slashbox}
\usepackage{balance}
\usepackage{pgf-umlsd}
\usepackage{enumerate}
\usepackage[shortlabels]{enumitem}

\usepackage{tikz}
\usetikzlibrary{arrows, decorations.markings}
\usetikzlibrary{shapes, shadows, arrows}
      \tikzstyle{vecArrow} = [thick, decoration={markings,mark=at position
      1 with {\arrow[semithick]{open triangle 60}}},
      double distance=1.4pt, shorten >= 5.5pt,
      preaction = {decorate},
      postaction = {draw,line width=1.4pt, white,shorten >= 4.5pt}]
     \tikzstyle{innerWhite} = [semithick, white,line width=1.4pt, shorten >= 4.5pt]
\tikzset{font=\small}
\usetikzlibrary{positioning}
\usepackage{tikz-cd}
\usetikzlibrary{arrows,positioning}
\tikzset{
  shift left/.style ={commutative diagrams/shift left={#1}},
  shift right/.style={commutative diagrams/shift right={#1}}
}

\usepackage{pgf-umlsd}
\usepackage{dblfloatfix}

\usepackage[bookmarks=false]{hyperref}

\usepackage{xcolor}

\newtheorem{definition}{Definition}
\newtheorem{theorem}{Theorem}
\newtheorem{myAttack}{Attack}
\newtheorem{proposition}{Proposition}
\usepackage{latexsym}
\usepackage{amssymb,amsmath}
\usepackage{amsmath}
\usepackage{breqn}
\usepackage{multicol}
\usepackage[skins]{tcolorbox}
\newtcolorbox{myframe}[2][]{%
  enhanced,colback=white,colframe=black,coltitle=black,
  sharp corners,boxrule=0.6pt,
  fonttitle=\itshape,
  attach boxed title to top left={yshift=-0.3\baselineskip-0.4pt,xshift=2mm},
  boxed title style={tile,size=minimal,left=0.5mm,right=0.5mm,
    colback=white,before upper=\strut},
  title=#2,#1
}

\usepackage{amsmath}
\usepackage{pifont}

\usepackage{algorithm}
\usepackage[noend]{algpseudocode}
\def\BibTeX{{\rm B\kern-.05em{\sc i\kern-.025em b}\kern-.08em
    T\kern-.1667em\lower.7ex\hbox{E}\kern-.125emX}}
    
\usepackage{tikz,pgfplots}    
    
\usepackage{tikz}
\usetikzlibrary{arrows, decorations.markings}
\usetikzlibrary{shapes, shadows, arrows}
      \tikzstyle{vecArrow} = [thick, decoration={markings,mark=at position
      1 with {\arrow[semithick]{open triangle 60}}},
      double distance=1.4pt, shorten >= 5.5pt,
      preaction = {decorate},
      postaction = {draw,line width=1.4pt, white,shorten >= 4.5pt}]
     \tikzstyle{innerWhite} = [semithick, white,line width=1.4pt, shorten >= 4.5pt]
\tikzset{font=\small}
\usetikzlibrary{positioning}
\usepackage{tikz-cd}
\usetikzlibrary{arrows,positioning}
\tikzset{
  shift left/.style ={commutative diagrams/shift left={#1}},
  shift right/.style={commutative diagrams/shift right={#1}}
}
\usepackage{lipsum,adjustbox}

\pgfplotsset{
compat=1.7,
every axis/.append style={
line width=1.25pt,
tick style={line width=1.25pt, color=black, line cap=round}
}
}

\definecolor{Gray}{gray}{0.9}
\definecolor{airforceblue}{rgb}{0.36, 0.54, 0.66}
\definecolor{aliceblue}{rgb}{0.94, 0.97, 1.0}
\definecolor{alizarin}{rgb}{0.82, 0.1, 0.26}
\definecolor{amber}{rgb}{1.0, 0.75, 0.0}
\definecolor{amber(sae/ece)}{rgb}{1.0, 0.49, 0.0}
\definecolor{applegreen}{rgb}{0.55, 0.71, 0.0}
\definecolor{babyblue}{rgb}{0.54, 0.81, 0.94}
\definecolor{ballblue}{rgb}{0.13, 0.67, 0.8}
\definecolor{beaublue}{rgb}{0.74, 0.83, 0.9}
\definecolor{bronze}{rgb}{0.8, 0.5, 0.2}
\definecolor{battleshipgrey}{rgb}{0.52, 0.52, 0.51}
\definecolor{bole}{rgb}{0.47, 0.27, 0.23}
\definecolor{bulgarianrose}{rgb}{0.28, 0.02, 0.03}
\definecolor{carolinablue}{rgb}{0.6, 0.73, 0.89}
\definecolor{ceil}{rgb}{0.57, 0.63, 0.81}
\definecolor{cerulean}{rgb}{0.0, 0.48, 0.65}
\definecolor{charcoal}{rgb}{0.21, 0.27, 0.31}
\definecolor{columbiablue}{rgb}{0.61, 0.87, 1.0}
\definecolor{coolblack}{rgb}{0.0, 0.18, 0.39}
\definecolor{darkcandyapplered}{rgb}{0.64, 0.0, 0.0}
\definecolor{darkbrown}{rgb}{0.4, 0.26, 0.13}
\definecolor{darkgray}{rgb}{0.66, 0.66, 0.66}
\definecolor{darkjunglegreen}{rgb}{0.1, 0.14, 0.13}
\definecolor{darktaupe}{rgb}{0.28, 0.24, 0.2}
\definecolor{frenchblue}{rgb}{0.0, 0.45, 0.73}
\definecolor{almond}{rgb}{0.94, 0.87, 0.8}
\definecolor{beige}{rgb}{0.96, 0.96, 0.86}
\definecolor{bisque}{rgb}{1.0, 0.89, 0.77}
\definecolor{black}{rgb}{0.0, 0.0, 0.0}
\definecolor{fluorescentorange}{rgb}{1.0, 0.75, 0.0}
\definecolor{ghostwhite}{rgb}{0.97, 0.97, 1.0}
\definecolor{antiquewhite}{rgb}{0.98, 0.92, 0.84}
\definecolor{davy\'sgrey}{rgb}{0.33, 0.33, 0.33}

\usepackage{xcolor}
\usepackage{lipsum}
\def\BibTeX{{\rm B\kern-.05em{\sc i\kern-.025em b}\kern-.08em
    T\kern-.1667em\lower.7ex\hbox{E}\kern-.125emX}}

\newcommand{\LotSFull}{Private Searchable Functional Encryption}
\newcommand{\LotS}{PSFE}

\begin{document}

\title{Heal the Privacy: Functional Encryption and Privacy-Preserving Analytics\\
%\thanks{This work was funded by the ASCLEPIOS: Advanced Secure Cloud Encrypted Platform for Internationally Orchestrated Solutions in Healthcare Project No. 826093 EU research project.}
}

%\IEEEpubid{978-1-7281-8086-1/20/\$31.00 \copyright2020 IEEE}

\author{
    \IEEEauthorblockN{Alexandros Bakas, Antonis Michalas}
    \IEEEauthorblockA{Tampere University
    \\\{alexandros.bakas, antonios.michalas\}@tuni.fi}
%    \IEEEauthorblockA{}
}
%\thanks{This work was funded by the ASCLEPIOS: Advanced Secure Cloud Encrypted Platform for Internationally Orchestrated Solutions in Healthcare Project No. 826093 EU research project.}
%
%\IEEEoverridecommandlockouts
%\IEEEpubid{\makebox[\columnwidth]{978-1-7281-8086-1/20/\$31.00~\copyright2020 IEEE \hfill} \hspace{\columnsep}\makebox[\columnwidth]{ }}

\maketitle

\IEEEpubidadjcol

\begin{abstract}
Secure cloud storage is an issue of paramount importance that both businesses and end-users should take into consideration before moving their data to, potentially, untrusted clouds. Migrating data to the cloud raises multiple privacy issues, as they are completely controlled by a cloud provider. Hence, an untrusted cloud provider can potentially breach users' privacy and gain access to sensitive information. The problem becomes even more pronounced when the cloud provider is required to store a statistical database and periodically publish analytics. In this work, we first present a detailed example showing that the use of cryptography is not enough to ensure the privacy of individuals. Then, we design a hybrid protocol based on Functional Encryption and Differential Privacy that allows the computations of statistics in a privacy-preserving way.
\end{abstract}

%\begin{abstract}
%This document is a model and instructions for \LaTeX.
%This and the IEEEtran.cls file define the components of your paper [title, text, heads, etc.]. *CRITICAL: Do Not Use Symbols, Special Characters, Footnotes, 
%or Math in Paper Title or Abstract.
%\end{abstract}


\section{Introduction}
\label{sec:intro}

Statistics, and data analytics in general, are very important tools for a variety of predictions. From  real-time traffic analysis to disease outbreaks discovery, statistics allow societies to predict critical situations and prepare accordingly. However, along with the growth of cloud computing, such prediction services are  moving to the cloud, where untrusted third parties may host and control statistical databases. Naturally, this raises several security concerns as the privacy of individuals can often be breached. These concerns become even greater when the analytics in question refer to extra sensitive data, such as medical records. A first response to these problems was presented in~\cite{dwork2006calibrating} with the formalization of differential privacy.

Differential Privacy allows sharing information about a dataset while withholding information about the individuals. In a differential private scheme, a curator (data owner) generates a dataset and, upon request of an analysts, releases statistics. To ensure the individuals' privacy the statistics are filtered through a privacy mechanism and finally the analyst receives a noisy result. However, the results must be published in a way that will allow any analyst to deduce accurate enough results without breaching the privacy of any given individual. Although the problem of privatizing datasets has been thoroughly studied by both researchers and big industrial players like Apple and Google, the problem of further securing datasets with encryption has not drawn much attention so far. In this work, we aim to combine differential privacy with the promising concept of Functional Encryption (FE) in an attempt to design a protocol for privacy-preserving release of statistics.

FE is an emerging cryptographic technique which allows computations over encrypted data. More precisely, FE schemes provide key generation algorithms that output decryption keys with remarkable capabilities. In contrast to traditional cryptography, each functional decryption key $\mathsf{sk}_f$ is associated with a function $f$. Decrypting a ciphertext $\mathsf{Enc}(x)$ using $\mathsf{sk}_f$, yields $f(x)$ and thus keeps the $x$ private. More recent works~\cite{goldwasser2014multi} generalized the concept of FE by presenting Multi-Input Functional Encryption (MIFE). In a MIFE scheme, given encryptions $\mathsf{Enc}(x_1), \dots, \mathsf{Enc}(x_n)$, a user can use $\mathsf{sk}_f$ to recover $f(x_1, \dots, x_n)$. In our work, we combine MIFE with differential privacy to design a scheme that allows the periodical release of statistics in a privacy-preserving way.

\bigskip

\paragraph*{\textbf{Contribution}} To the best of our knowledge this is amongst the first works that combine differential privacy with cryptography to ensure the security of datasets, and the first one that does so using FE. More specifically:

\begin{enumerate}[label=\bfseries C\arabic*.]
	\item By combining FE with differential privacy, we propose a hybrid protocol as solution to the problem of designing encrypted private databases. Our work draws inspiration from both the fields of FE and Symmetric Searchable Encryption~\cite{10.1007/978-3-030-81242-3_5,9219739}
	\item We provide a detailed security analysis of our protocol by demonstrating that it remains secure in the presence of a malicious adversary. Furthermore, we formally prove that our protocol satisfies the notion of differential privacy.
	\item Our solution is considered as efficient since it relies only on symmetric cryptographic primitives.%Our solution relies only on symmetric cryptographic primitives that are characterized by their efficiency.
\end{enumerate}

%While data analytics start dominating several fields, there are a lot of concerns regarding the use of users' data by companies and how this can breach their privacy.

\bigskip

\paragraph*{\textbf{Organization}} The rest of the paper is organised as follows: In \autoref{sec:motivation}, we present a concrete example that proves that cryptography is not enough to secure statistical databases. In \autoref{sec:related}, we discuss important published works in the fields of functional encryption and differential privacy. \autoref{sec:background} contains all the necessary notations, cryptographic primitives and security notions used throughout the paper, and is followed by \autoref{sec:architecture}, where we present the detail of our system model. Section~\ref{sec:formal} demonstrates the core contribution of the work as we present a scheme for publishing statistics in a privacy-preserving way. The security of our construction is proved in \autoref{sec:security} and finally, \autoref{sec:conclusions} concludes the paper.

\section{Motivation and Application Domain}
\label{sec:motivation}

The ultimate goal of this work is to enable authorized users (analysts) to perform statistical analyses over medical datasets in a privacy-preserving way. In order to make this possible, we needed to ensure that our construction would be resistant against both internal (e.g.\ malicious servers) as well as external (e.g.\ malicious analysts) attacks.  

For our solution, we used structured datasets composed of three different kinds of variables: \textit{categorical,} \textit{ordinal} and \textit{numerical}:

%In our solution, we made use of structured datasets in which we distinguished between three different kinds of variables; categorical, ordinal and numerical:

\begin{enumerate}[\bfseries a.] 
	\item \textbf{Categorical variables} do not have a natural ordering. For example, the medical diagnosis of a patient is a categorical variable.
	\item \textbf{Ordinal variables} are categorical variables for which possible values can be ordered. For example, the condition of a patient, for which we can arbitrarily assume to be \texttt{mild} $<$ \texttt{severe} $<$ \texttt{critical} is considered an ordinal variable.
	\item \textbf{Numerical variables} are expressed using numbers e.g.\ the age of a patient or systolic and diastolic blood pressures.
\end{enumerate}

\medskip

To make things clearer, let us consider a scenario in which four patients sought medical care. Following the examination, the hospital stores their medical records to a structured dataset. As a next step, the hospital (who in this case acts as the curator of the dataset) masks all the ordinal and categorical variables in the dataset using a cryptographic hash function, and encrypts the numerical variables with a MIFE scheme. Without loss of generality, we can assume that our dataset looks like the one in~\autoref{tab:dataset}, where each $H(\cdot)$ denotes the hash of a variable and each $c_x$ denotes the ciphertext corresponding to a plaintext $x$.  Finally, the dataset is outsourced to a cloud service provider (CSP), where it will be stored. 

 \begin{table}[t]
 \centering
 \scalebox{0.8}{
 \begin{tabular}{ |p{1.5cm}||p{2.0cm}||p{1.7cm}||p{1.0cm}||p{1.0cm}||p{1.0cm}|}
 \hline
 \rowcolor{davy\'sgrey}
 \multicolumn{6}{|c|}{\color{white}{\textbf{Dataset}}} \\
 \hline
 \rowcolor{Gray}
 \textbf{Patients}     & \textbf{Diagnosis} & \textbf{Condition}  & \textbf{Age} & \textbf{sbp} & \textbf{dbp} \\
 \hline
 $H(\mathrm{Dennis})$ &   $H(\mathrm{covid}19)$ &  $H(\mathrm{mild})$   & $c_{27}$ &  $c_{110}$ & $c_{75}$  \\ 
 \hline 
 $H(\mathrm{Shawn})$ &   $H(\mathrm{flu})$ &   $H(\mathrm{severe})$ &  $c_{58}$ &    $c_{123}$ &   $c_{60}$  \\
 \hline
 $H(\mathrm{Dirk})$ &   $H(\mathrm{flu})$ &   $H(\mathrm{mild})$ &   $c_{41}$ &   $c_{120}$ &   $c_{80}$ \\
 \hline
 $H(\mathrm{Scottie})$ &   $H(\mathrm{pneumonia})$ &   $H(\mathrm{critical})$ &   $c_{65}$ &   $c_{149}$ &   $c_{58}$ \\
 \hline
\end{tabular}}
\bigskip
\caption{Structured dataset with four cases and five variables. sbp: systolic blood pressure, dbp: diastolic blood pressure.}
%\caption{Comparison of our construction to other important works. The notation is as follows: $n$: dataset size, $\ell$: result size, $R$: range, $p$: processors  size}
\label{tab:dataset}
\end{table}

For our work, we want to enable analysts to query the CSP in a privacy-preserving manner with queries in the form of ``\textit{What is the average age of all patients that have been diagnosed with covid19?}" or \textit{``What is the blood pressure of the patients whose condition is severe?"}. In other words, we want to be able to compute a function on the values of the numerical variables that correspond to a specific categorical or ordinal variable. 

It should be noted that although cryptography ensures the data confidentiality, it does not ensure the individuals' privacy. For example, if an analyst were to initially requests the average age for the first three cases in the dataset and subsequently request the average age of all patients, it would become obvious how Scottie's age influences the average and hence, its value could be deduced. To protect the individuals' privacy, we rely on the notion of differential privacy. By embedding well-calibrated error in the decryption algorithm, we ensure that the analyst has access to accurate enough results in order to perform any kind of analytics, without breaching the induvidual's privacy.

\section{Related Work}
\label{sec:related}

\paragraph*{\textbf{Functional Encryption}} While numerous studies with general definitions and generic constructions of FE have been proposed~\cite{boneh2012functional,10.1145/3133956.3134106,cryptoeprint:2021:1692,10.1145/3508398.3511514,10.1007/978-3-030-70852-8_7,waters2015punctured} there is a clear lack of works proposing FE schemes supporting specific functions. To the best of our knowledge, currently the number of supported functionalities is limited to inner products~\cite{10.1007/978-3-319-96884-1_20,abdalla2017multi,abdalla2015simple}, quadratic polynomials~\cite{sans2018reading} and the $\ell_1$ norm of a vector~\cite{9342992}. In this work, we use the symmetric construction for the $\ell_1$ norm presented in~\cite{9342992} to design a functionally encrypted private scheme. 

\paragraph*{\textbf{Differential Privacy}} Differential privacy is a notion first formalized in~\cite{dwork2006calibrating}, where authors focused on ensuring the individuals' privacy. More precisely, they proved that by adding well-calibrated noise to the data, the presence or absence of an individual's information is \textit{irrelevant} to the output of a database query. Since then, differential privacy has drawn the attention of both researchers~\cite{andres2013geo,barak2007privacy,blum2013learning,kasiviswanathan2011can}
and key industry players such as Google~\cite{fanti2016building} and Uber~\cite{johnson2018towards}. Nonetheless, to the best of our knowledge the only work that combines differential privacy with cryptography is the one presented in~\cite{agarwal2019encrypted}, where authors designed a scheme for private histogram queries. However, the solution presented in~\cite{agarwal2019encrypted} relied on homomorphic encryption and hence, queries were restricted to only asking for the value of a counter. In our work, by using FE we allow users to perform any kind of query that is supported by the functionality of the FE scheme.

\section{Background}
\label{sec:background}

\paragraph*{\textbf{Notation}} If $\mathcal{Y}$ is a set, we use $y \xleftarrow{\$} \mathcal{Y}$ if $y$ is chosen uniformly at random from $\mathcal{Y}$. The cardinality of a set $\mathcal{Y}$ is denoted by $|\mathcal{Y}|$. Vectors are denoted in bold as  $\mathbf{x} = [x_1, \ldots ,x_n]$.   A probabilistic polynomial time (PPT) adversary $\mathcal{ADV}$ is a randomized algorithm for which there exists a polynomial $p(z)$ such that for all input $z$, the running time of $\mathcal{ADV}(z)$ is bounded by $p(|z|)$.

\subsection{Functional Encryption}
\label{subsec:FE}

\begin{definition}[Multi-Input Functional Encryption in the Symmetric Key Setting]
Let $\mathcal{F} = \{f_1, \ldots, f_n\}$ be a family of n-ary functions where each $f_i$ is defined as follows: ${f_i : \mathbb{Z}^n \rightarrow \mathbb{Z}}$. A multi-input functional encryption scheme for $\mathcal{F}$ consists of the following algorithms:

\begin{itemize}
	\item $\mathsf{Setup}(1^\lambda):$ Takes as input a security parameter $\lambda$  and outputs a secret key $\mathbf{K} = [\mathsf{k}_1, \ldots ,\mathsf{k}_n] \in \mathbb{Z}^n$.
	\item $\mathsf{Enc(\mathsf{K}}, i, x_i):$ Takes as input $\mathbf{K}$, an index $i \in [n]$ and a message $x_i \in \mathbb{Z}$ and outputs a ciphertext $ct_i$.
	\item $\mathsf{KeyGen(\mathbf{K}}, f):$ Takes as input $\mathbf{K}$  and a description of a function $f_i$ and outputs a functional decryption key $\mathsf{sk}_{f_i}$.
	\item $\mathsf{Dec}(\mathsf{sk}_{f_i}, ct_1, \dots, ct_n):$ Takes as input a decryption key $\mathsf{sk}_{f_i}$ for a function $f_i$ and $n$ ciphertexts and outputs a value $y \in \mathbb{Z}$.
\end{itemize}
\end{definition}

For the needs of our work, we rely on the one-AD-IND-secure symmetric MIFE scheme for the $\ell_1$ norm, presented in~\cite{9342992}. Informally, one-AD-IND security ensures that given the encryption of two messages $x_1$ and $x_2$, and a functional key $\mathsf{sk}_f$ for a function $f$ such that $f(x_1) = f(x_2)$, no PPT adversary should be able to distinguish between them. With the aim of completeness and improved readability, the MIFE scheme for the $\ell_1$ norm is illustrated in~\autoref{fig:MIFEl}.

\begin{figure}[H]
\begingroup
\fontsize{9.2pt}{4pt}\selectfont
\begin{myframe}{}
\begin{multicols}{2}
\setlength\columnsep{100pt}
	 \underline{$\mathsf{Setup}(1^\lambda):$}\\
			$\forall\: i\in[n], \mathsf{k}_i \xleftarrow{\$} \mathbb{Z}$\\ 
			 Return $\mathbf{K} = [\mathsf{k_1, \dots, k_n}] \in \mathbb{Z}^n$\\
			 %Output the public parameters $PP = (n, \|\mathbf{s}\|_1)$\\
		
	 \underline{$\mathsf{Enc(\mathbf{\mathbf{K}}, i , x_i)}:$}\\
				 Return  $ct_i =  x_i + \mathsf{k}_i$

\columnbreak		
		
	 \underline{$\mathsf{KeyGen}(\mathbf{K}):$}\\
			 Return $\mathsf{sk}_f = \|\mathsf{K}\|_1 = \sum_{i}^{n}\mathsf{k}_i$\\
								
	 \underline{$\mathsf{Dec}(\mathsf{sk}_f, ct_1, \dots, ct_n):$}\\
		 Return $\sum_{i=1}^{n} ct_i - \mathsf{sk}_f$
		 \end{multicols}
\end{myframe}
\endgroup
\caption{one-AD-IND-secure MIFE for the $\ell_1$ norm ($\mathsf{MIFE}_{\ell_1}$).}
\label{fig:MIFEl}
\end{figure}

\subsection{Differential Privacy}
\label{subsec:DP}

We proceed by providing the main definitions of $\epsilon$-differential privacy ($\epsilon$-DP) and the main properties of the Laplace mechanism. 
%We briefly recall the main definitions of $\epsilon$-differential privacy ($\epsilon$-DP) and the main properties of the Laplace mechanism. 

%\begin{definition}[$\ell_1$-distance]
%The $\ell_1$-distance between two datasets DS and DS' is given by:
%\begin{dmath}
%\|DS-DS'\|_1
%\end{dmath}
%The $\ell_1$-distance counts the number of entries on which the two datasets differ.
%\end{definition}

\begin{definition}
Two datasets DS and DS' are neighbouring if:
\begin{dmath}
\|DS-DS'\|_1 \leq 1
\end{dmath}
\end{definition}

\begin{definition}[$\epsilon$-DP]
\label{def:DP}
 A privacy mechanism $\mathcal{M}: \mathbb{N}^{|DS|} \rightarrow Im(\mathcal{M})$ is $\epsilon$-DP if $\forall \, \mathcal{S}\subset Im(\mathcal{M})$ and $\forall$ neighboring datasets ${DS, DS' \in \mathbb{N}^{|DS|}:}$

\begin{dmath*}
{Pr[\mathcal{M}(DS) \in \mathcal{S}] \leq e^\epsilon Pr[\mathcal{M}(DS) \in \mathcal{S}]}
\end{dmath*}

\end{definition}

%\begin{definition}[$\ell_1$ sensitivity]
%The $\ell_1$ sensitivity of a query $q: \mathbb{N}^{|DS|} \rightarrow \mathbb{R}$ is: 
%\begin{dmath*}
%\Delta q = \max_{d(DS, DS') \leq 1}\|q(DS) - q(DS')\|_1
%\end{dmath*}
%\end{definition}
\begin{definition}[Laplace distribution] The Laplace distribution centered at 0 and with scale parameter $b$ is given by:

\begin{dmath*}
Lap(z) = \frac{1}{2b}e^{-\frac{|z|}{b}}
\end{dmath*}

where the mean is 0 and the variance is $2b^2$.
 
\end{definition}

We are now ready to proceed with the definition of the Laplace Mechanism~\cite{dwork2006calibrating}.

\begin{definition}[Laplace Mechanism]
Given a query ${q: \mathbb{N}^{|\mathcal{D}|} \rightarrow \mathbb{R}}$, the Laplace Mechanism is:
\begin{dmath*}
M_L(DS, q, \epsilon) = q(DS) + Y_i,
\end{dmath*}
where $Y_i \sim Lap(b)$
\end{definition}

A proof showing that the Laplace Mechanism is  $\epsilon$-differentially private can be found in~\cite{dwork2006calibrating}.

\section{Architecture}
\label{sec:architecture}
In this section, we introduce the system model by explicitly describing the main entities participating in our protocol along with their capabilities. 

We assume the existence of the following four entities:

\begin{enumerate}[\bfseries 1.]
	\item \textbf{Curator (\textbf{C}):} \textbf{C} is responsible for generating an encrypted dataset and outsourcing to the CSP. \textbf{C} also generates a list $L_{MA}$ containing mappings between encryption keys and their unique identifiers. This list is outsourced to MA.
	\item {\textbf{Analyst (\textbf{A})}:} \textbf{A} is an analyst that can perform statistics on the data stored in the CSP. 
	\item {\textbf{Cloud Service Provider (CSP)}:} We consider a cloud computing environment based on a trusted IaaS provider 
similar to the one described in~\cite{Michalas:17:Trusted:Launch}. The CSP is responsible for storing an encrypted dataset. Apart from that, upon \textbf{A}'s request the CSP is required to perform a search operation on the encrypted dataset and further communicate with the Master Authority for the generation of secret functional keys.
	\item {\textbf{Master Authority (MA)}:} \textbf{MA} is a trusted authority that is responsible for issuing secret functional keys. To do so, MA is required to maintain a list containing mappings between encryption keys and their unique identifiers. 
\end{enumerate}

\section{Formal Construction}
\label{sec:formal}

This Section presents the core contribution of this work as we formally present \LotSFull(\LotS). We assume the existence of an IND-CCA2 secure public key cryptosystem and a EUF-CMA secure signature scheme. Finally, we also utilize a first and second preimage resistant hash function $H$. $\mathsf{\LotS}$ consists of three algorithms  $\mathsf{Gen, Setup}$ and $\mathsf{Read}$ such that:

\noindent{\textbf{\LotS.Gen:}} Each entity from the described architecture receives a public/private key pair $\mathsf{(pk,sk)}$ for an IND-CCA2 secure public cryptosystem, and publishes its public key while keeping the private key secret. Apart from that, all entities generate a signing and a verification key. Below we provide a list of all the generated keys:

\begin{itemize}
	\item $(\mathsf{pk_C, sk_C}), (\mathsf{sig_C, ver_C})$ - public/private, signing/verification and MIFE secret key for the Curator; 
	\item $(\mathsf{pk_A, sk_A}), (\mathsf{sig_A, ver_A})$ - public/private and signing/verification key pairs for the Analyst;
	\item $(\mathsf{pk_{CSP}, sk_{CSP}}), (\mathsf{sig_{CSP}, ver_{CSP}})$ - public/private and signing/verification key pairs for the cloud service provider;
	\item $(\mathsf{pk_{MA}, sk_{MA}}), (\mathsf{sig_{MA}, ver_{MA}})$ - public/private, signing/verification key pairs for the master authority.
\end{itemize} 

\noindent{\textbf{\LotS.Setup:}} Represents a three party protocol between \textbf{C}, the CSP and MA. $\mathsf{\LotS.Setup}$ is initiated by $\mathbf{C}$ who wants to outsource an encrypted dataset (EDS) to the CSP. To encrypt the dataset, \textbf{C} hashes all the categorical and the ordinal entries concatenated with a salt $s$ to prevent dictionary attacks. Apart from that, \textbf{C} also hashes the entries without the salt and stores each pair (salted and unsalted hashed entry) in a list $L_{MA_{s}}$. For the numerical ones, \textbf{C} generates a symmetric key $\mathsf{k}$ and uses it to encrypt the corresponding entry. Apart from that, for each generated $\mathsf{k}$, \textbf{C} generates a unique index. The keys, along with their indexes are stored in a list $L_{MA_{k}}$. Finally, \textbf{C} sends $m_1 = \langle t_1, \mathsf{pk}_{CSP}(\mathrm{EDS}), \sigma_C(H(t_1||EDS))\rangle$ to the CSP and $m_2 = \langle t_2, \mathsf{pk_{MA}}(L_{MA_{k}}, L_{MA_{s}}), \sigma_C{H(t_2\|L_{MA})}\rangle$ to MA. Upon receiving these messages, both the CSP and MA verify their freshness (by looking at the timestamps $t_1$ and $t_2$) and the identity of the sender (by verifying the signature). If the verifications are successful, the CSP stores EDS and MA stores both $L_{MA_{s}}$ and $L_{MA_{k}}$. In addition to that, both the CSP and MA send an acknowledgement to \textbf{C} that they have successfully stored EDS and the two lists via $m_3 = \langle t_3, \sigma_{CSP}(H(t_3\|EDS)) \rangle$ and $m_4 = \langle t_4,  \sigma_{MA}(t_4\|L_{MA_{s}}\|L_{MA_{s}}) \rangle$ respectively. The encryption of the dataset is presented in detail in algorithm~\autoref{alg:dataset encryption} and the flow of $\mathsf{PSFE.Setup}$ is illustrated in~\autoref{fig:setup}.

\begin{algorithm}[H]
\caption{Dataset Encryption}
\label{alg:dataset encryption}
\footnotesize
\begin{algorithmic}[1]
%\Statex \textbf{Input: $\mathsf{K}, \mathbf{f}$}
%\Statex \textbf{Output:}  ($\mathsf{In_{CSP}}, \mathbf{c}), \mathsf{In_{TA}} \leftarrow \mathsf{Index(K}, \mathbf{f})$ 
\State \textbf{Input:} A plaintext Dataset DS
\State \textbf{Output:} An encrypted Dataset EDS
\State  $\mathbf{K} = \{\}$
\State $L_{MA} = \{\}$ 
\State [r, c] = size(DS) \Comment{Number or rows and columns}
 	\For {i = 1 to r} \Comment{All the cases}
 		\For {j = 1 to c} \Comment{All the variable}
 			\If {DS(i, j) == categorical OR ordinal}
 				\State $s_{i, j} \leftarrow \mathbb{Z}$
 				\State $L_{MA_{s}} = L_{MA_{s}} \cup H(\mathrm{DS(i,j)})\|(H(\mathrm{DS}(i,j))\|s_{i,j})$
 				\State DS(i, j) = $H(\mathrm{DS}(i, j)\|s_{i,j})$
 				
 			\Else 
 				\State Generate $\mathsf{k}_{i, j} \in \mathbb{Z}$
 				\State $\mathrm{index_{k_{i, j}}}$ = $H(\mathsf{k}_{i, j})$
 				\State DS(i, j) = (DS(i, j) + $\mathsf{k}_{i, j}) || \mathrm{index_{k_{i, j}}}$
 				\State $L_{MA} = L_{MA} \cup (\mathsf{k}_{i, j}||\mathrm{index_{k_{i, j}}})$
 			\EndIf
 		\EndFor
 	\EndFor
\State Outsource $L_{MA_{s}}$ and $L_{MA_{k}}$ to MA 		
\State EDS = DS	
\end{algorithmic}
\end{algorithm}

\begin{figure}[H]
  \centering
  \scalebox{0.7}{
  \begin{sequencediagram}
    \newthread{A}{Curator}{}
    \newinst[4.7]{B}{CSP}{}
    \newinst[0.8]{C}{MA}{}
    %\newinst[4.7]{D}{TEE}{}
    \begin{callself}{A}{Encrypt DS using Algorithm~\ref{alg:dataset encryption}}{EDS}
	\end{callself}
    \begin{messcall}{A}{$m_1 = \langle t_1, \mathsf{pk}_{CSP}(\mathrm{EDS}), \sigma_C(H(t_1\|EDS))\rangle$}{B}{}
     \end{messcall}
    \begin{messcall}{A}{$m_2 = \langle t_2, \mathsf{pk_{MA}}(L_{MA_{k}}, L_{MA_{s}}), \sigma_C{H(t_2\|L_{MA})}\rangle$}{C}
   \end{messcall}
    \begin{messcall}{B}{$m_3 = \langle t_3, \sigma_{CSP}(H(t_3\|EDS)) \rangle$}{A}
    \end{messcall}
    \begin{messcall}{C}{$m_4 = \langle t_4,  \sigma_{MA}(t_4\|L_{MA_{s}}\|L_{MA_{s}}) \rangle$}{A}
     \end{messcall}
   
  \end{sequencediagram}}
  \caption{Flow of $\mathsf{PFSE.Setup}$}
  \label{fig:setup}
\end{figure}

\noindent{\textbf{\LotS.Read:}} Represents a tree party protocol between \textbf{A}, the CSP and MA. $\mathsf{\LotS.Read}$ is initiated by the analyst $\textbf{A}$ wishing to perform statistical analysis on the encrypted dataset. To do so, \textbf{A} first generates a search token $\tau_s$ as $\tau_s = \langle H(w_i), H(w_j), f \rangle$, where $H(w_i)$ refers to a categorical or ordinal value, $H(w_j)$ refers to a variable, and $f$ is the description of a function that will be applied to the ciphertexts. Then, \textbf{A} sends $m_5 = \langle t_5, \tau_s, \sigma_A(H(t_5\|\tau_s)) \rangle$ to the MA. Upon reception, MA verifies the freshness and the signature of $m_5$. If the verification is successful, MA retrieves the list $L_{MA_{s}}$, containing the salted hashes, finds which salted values correspond to $H(w_i)$ and $H(w_j)$ and sends them to the CSP via $m_6 = \langle t_6, H(w_i\|s_i), H(w_j\|s_j), \sigma_{MA}(t_6\|H(w_i\|s_i)\|H(w_j\|s_j))$. Upon reception, the CSP verifies the freshness and the signature of $m_6$. If the verification is successful, the CSP finds the ciphertexts that correspond to $H(w_i\|s_i)$ with attribute $H(w_j\|s_j|)$ and sends the result $R$  back to \textbf{A} via $m_7 = \langle t_7, R, \sigma_{CSP}(H(t_6\|R))\rangle$. At the same time CSP retrieves the unique index for each ciphertext $c_i \in R$, and stores them in a list $L_{index}$ before outsourcing them to MA via $m_8 = \langle t_8, L_{index}, f, \sigma_{CSP}(H(t_7\|L_{index}\|f))\rangle$. Upon reception of $m_8$ (and if the verifications are successful), MA can construct the functional key $\mathsf{sk}_f$ as a linear combination of all the keys $\mathsf{k}_i$ such that $H(\mathsf{k}_i) \in L_{index}$. Apart from that, MA samples an error $e \approx Lap(1/\epsilon)$ and computes a noisy key $\mathsf{sk}_f' = \mathsf{sk}_f + \epsilon$. Finally, $sk_f'$ is send back to \textbf{A} via $m_9 = \langle t_9, \mathsf{pk_{A}}(\mathsf{sk_f'}), \sigma_{MA}(H(t_8\|\mathsf{sk_f}'))\rangle$. Upon receiving $\mathsf{sk}_f'$, \textbf{A} computes the result as follows:

\begin{dmath*}
\sum_{i=1}^{n}c_i - \mathsf{sk}_f = \sum_{1}^{n}(\mathsf{k}_i + x_i) - \sum_{1}^{n}\mathsf{k}_i + e = \sum_{1}^{n}x_i + e
\end{dmath*} 

$\mathsf{\LotS.Read}$ is illustrated in~\autoref{fig:read}.

\begin{figure}[H]
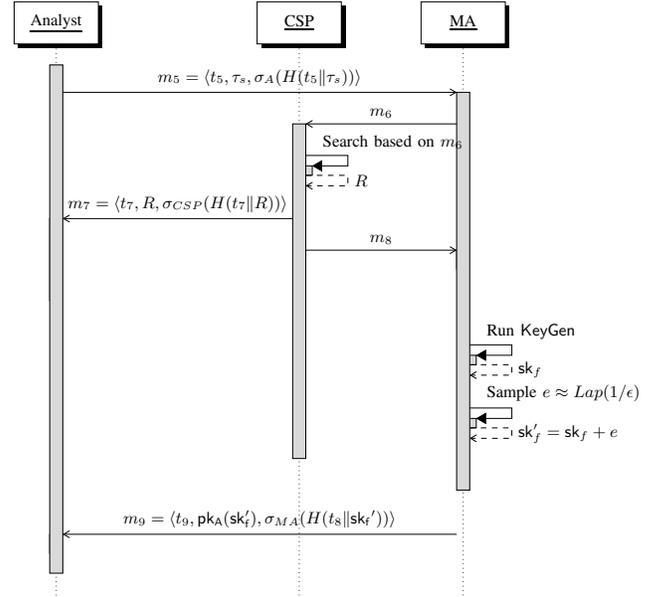

  \centering
  \scalebox{0.7}{
  \begin{sequencediagram}
    \newthread{A}{Analyst}{}
    \newinst[3.0]{B}{CSP}{}
    \newinst[1.5]{C}{MA}{}
    %\newinst[4.7]{D}{TEE}{}
    \begin{messcall}{A}{$m_5 = \langle t_5, \tau_s, \sigma_A(H(t_5\|\tau_s)) \rangle$}{C}{}
    \begin{messcall}{C}{$m_6$}{B}{}
    \begin{callself}{B}{Search based on $m_6$}{$R$}
	\end{callself}
    \begin{messcall}{B}{$m_7 = \langle t_7, R, \sigma_{CSP}(H(t_7\|R))\rangle$}{A}
    \begin{messcall}{B}{$m_8$}{C}{}    
     \end{messcall}
     \end{messcall}
     \begin{callself}{C}{Run $\mathsf{KeyGen}$}{$\mathsf{sk}_f$}
	\end{callself}
     \begin{callself}{C}{Sample $e \approx Lap(1/\epsilon)$}{$\mathsf{sk}_f' = \mathsf{sk}_f + e$}
	\end{callself}
	\end{messcall}
	 \end{messcall}
    \begin{messcall}{C}{$m_9 = \langle t_9, \mathsf{pk_{A}}(\mathsf{sk_f'}), \sigma_{MA}(H(t_8\|\mathsf{sk_f}'))\rangle$}{A}{}
    \end{messcall}
   
  \end{sequencediagram}}
  \caption{A complete run of $\mathsf{\LotS.Read}$ (The description of $m_6$ and $m_7$ is omitted from the diagram due to space constraints -- $m_6 = \langle t_6, H(w_i\|s_i), H(w_j\|s_j), \sigma_{MA}(t_6\|H(w_i\|s_i)\|H(w_j\|s_j))$ and ${m_8 = \langle t_8, L_{index}, f, \sigma_{CSP}(H(t_8\|L_{index}\|f))\rangle}$).}
  \label{fig:read}
\end{figure}

\section{Security Analysis}
\label{sec:security}

In this Section we prove the security of our protocol, and show that the $\mathsf{\LotS.Read}$ protocol is $\epsilon$-differential private one. Before proceeding to do so, we formally define our threat model.

\subsection{Threat Model}
\label{subsec:threat}

\paragraph*{\textbf{Threat Model}}Our threat model is similar to the one described in~\cite{Michalas:17:Trusted:Launch}, based on the Dolev-Yao adversarial model~\cite{dolev:1983}. We additionally  extend it by defining a set of new attacks.

\begin{myAttack}[Result Substitution Attack]
Let $\mathcal{ADV}$ be an adversary that observes the communication channels between \textbf{A} and the CSP. $\mathcal{ADV}$ successfully launches a \textit{Result Substitution Attack}, if she manages to replace the result list $R$, sent from the CSP to \textbf{A}, with another one $R'$.
\end{myAttack}

\begin{myAttack}[Key Substitution Attack]
Let $\mathcal{ADV}$ be an adversary that observes the communication channels between \textbf{A}, the CSP and MA. $\mathcal{ADV}$ successfully launches a \textit{Key Substitution Attack}, if $A$ receives a wrong $\mathsf{sk}_f'$ in a way that is indistinguishable to her. 
\end{myAttack}

\subsection{Protocol Security}
\label{subsec:protocol security}

We will proceed to prove $\mathsf{\LotS}$'s soundness against the attacks defined in Section~\ref{subsec:threat}.

\begin{proposition}[Result Substitution Attack Soundness]
\label{proposition:RSA}
Let $\mathcal{ADV}$ be an adversary that overhears the communication between \textbf{A} and the CSP. Then $\mathcal{ADV}$ cannot successfully launch a Result Substitution Attack.
\end{proposition} 

\begin{proof}
For $\mathcal{ADV}$ to successfully launch a Result Substitution Attack, she needs to tamper with the result list $R$ that is sent from the CSP to \textbf{A} via $m_6 = \langle t_6, R, \sigma_{CSP}(H(t_6\|R))\rangle$. To do, $\mathcal{ADV}$ has two choices:

\begin{itemize}
	\item Reply an old $m_6$ message
	\item Replace $R$ with another result list $R_{mal}$
\end{itemize}
\medskip

In the instance where $\mathcal{ADV}$ overhears the communication between \textbf{A} and the CSP, we can assume that $\mathcal{ADV}$ possesses an old $m_6$ message $m_{6_{old}} = \langle t_{6_{old}}, R_{old}, \sigma_{CSP}(H(t_{6_{old}}\|R_{old}))\rangle$. Thus, when the CSP sends $m_6$ to \textbf{A}, $\mathcal{ADV}$ intercepts the communication and replaces $m_6$ with $m_{6_{old}}$. Upon receiving $m_{6_{old}}$, \textbf{A} verifies the signature, and since $m_{6_{old}}$ contains a valid CSP's signature, the verification is successful. However, when \textbf{A} tries to verify the freshness of the message, she notices that the timestamp is old and thus drops the communication. As a result, the only way for $\mathcal{ADV}$ to successfully launch the attack, is to use another result list $R_{mal}$.

Just like before, when the CSP sends $m_6$ to \textbf{A}, $\mathcal{ADV}$ intercepts the communication and replaces $R$ with $R_{mal}$. However, the result list $R$ is also included in the CSP's signature. Thus, replacing $R$ with $R_{mal}$ in an indistinguishable way, is equivalent to forging the CSP's signature. Nonetheless, given the signature scheme's EUF-CMA security, there is only a negligible probability for this to happen and hence, the attack fails.

\end{proof}

\begin{proposition}[Key Substitution Attack Soundness]
\label{proposition:KSA}
Let $\mathcal{ADV}$ be an adversary that overhears the communication channels between \textbf{A}, the CSP and MA. Then \textbf{A} cannot successfully launch a Key Substitution Attack.
\end{proposition}

\begin{proof}
Since the encryption keys are elements in $\mathbb{Z}$, and the secret functional key is a linear combination of the encryption keys, it follows that the functional key lives in $\mathbb{Z}$ as well. This means that \textbf{A} is expecting to receive an integer and hence $\mathcal{ADV}$ could easily replace the real integer number in a way that is indistinguishable for $\textbf{A}$. As a result, we need to make sure that even the slightest modification in the structure of the messages will have a big impact on what \textbf{A} receives.
For $\mathcal{ADV}$ to successfully launch a Key Substitution Attack, she needs to replace the key $\mathsf{sk}_f'$ with one of her choice in a way that is indistinguishable for \textbf{A}. To do so, $\mathcal{ADV}$ can follow the following two approaches:

\begin{itemize}
	\item Replace the functional key sent from MA to \textbf{A}, as part of $m_8$, with a key of her choice.
	\item Force MA to compute a functional key for a function $g$ such that $g \neq f$.
\end{itemize}

\medskip

Tampering with the $m_8$ message sent from MA to \textbf{A} requires $\mathcal{ADV}$ to either use an old $m_8$ message or forge the signature of MA. As we saw in the proof for proposition~\autoref{proposition:RSA}, the use of timestamps and the EUF-CMA security of the signature scheme, ensure that $\mathcal{ADV}$ can only achieve this with negligible probability. Hence, we conclude that the only way for $\mathcal{ADV}$ to successfully launch a Key Substitution Attack is to force MA to compute a functional key for a function $g \neq f$.

Fooling MA into computing a wrong functional key, requires $\mathcal{ADV}$ to tamper with the $m_7$ message sent from the CSP to MA. Recall that $ \langle t_7, L_{index}, f, \sigma_{CSP}(H(t_7\|L_{index}\|f))\rangle$. By observing the structure of $m_7$, we see that $\mathcal{ADV}$ can either target $L_{index}$, the description of the function $f$, or both. However, similarly to the proof for proposition~\autoref{proposition:RSA}, as $L_{index}$ and the function $f$'s description are also included in the CSP's signature, tampering with them is equivalent to forging the CSP's signature, which can only happen with negligible probability due to the signature scheme's  EUF-CMA security. Moreover, the timestamp, ensures that $\mathcal{ADV}$ cannot replace an older message. We thus prove that in both cases, $\mathcal{ADV}$ can successfully launch a Key Substitution Attack with negligible probability.
\end{proof}
\subsection{Differential Privacy}

In this section, we prove that the $\mathsf{\LotS}.Read$ protocol is $\epsilon$-differential private one.

\begin{theorem}
Let $EDS, EDS' \in \mathbb{N}^{|DS|}$ be arbitrary neighbouring datasets, let $q: \mathbb{N}^{|DS|} \rightarrow \mathbb{R}$ be an arbitrary query and let $r, r'\in \mathbb{R}$. Moreover, let $M_L$ be the Laplace Mechanism. Then, the $\mathsf{\LotS.Read}$ protocol is $\epsilon$-differentially private as per Definition~\ref{def:DP}.
\end{theorem}

\begin{proof}

Our goal is to prove that issuing a private query $q$ to EDS reveals no more information than what is allowed by the privacy factor $\epsilon$. In our construction, when \textbf{A} uses $sk_f'$ to decrypt the result list $R$, she gets a result $r'$. In other words, $q(EDS) = r'$. However, the query contains the secret functional key $\mathsf{sk}_f' = \mathsf{sk}_f + e$, where $e \leftarrow Lap(\frac{\Delta q}{\epsilon})$, and when the Laplace Mechanism is applied to the query we get:

\begin{dmath*}
{M_L(q, EDS, \epsilon) = r = r' + e \Rightarrow} \\ {e = r - r' \Rightarrow} \\ e = r - q(EDS).
\end{dmath*}

However, since $e$ is arbitrarily chosen from the Laplace distribution ($e \xleftarrow{\$} Lap(\frac{\Delta q}{\epsilon})$), then, without loss of generality, we can replace $e$ with $Lap(\frac{\Delta q}{\epsilon})$. Hence, we get:

\begin{dmath*}
{\frac{Pr[M_L(EDS, q, \epsilon) = r]}{Pr[M_L(EDS', q, \epsilon) = r]} = \frac{Pr[Lap(\frac{\Delta q}{\epsilon}) = r - q(EDS)]}{Pr[Lap(\frac{\Delta q}{\epsilon}) = r - q(EDB')]} }={ \frac{\frac{\epsilon}{2\Delta q}exp{(-\frac{|r-q(EDS)|}{\Delta q}\epsilon)}}{\frac{\epsilon}{2\Delta q}exp{(-\frac{|r-q(EDS')|}{\Delta q}\epsilon)}}}   = {exp\left(-\frac{\epsilon |r - q(EDS')| - |r -q(EDS)|)}{\Delta q}\right)} = {exp\left(\frac{\epsilon |q(EDS') - q(EDS)|}{\Delta q} \right) \leq e^\epsilon} 
\end{dmath*}
\end{proof}

\bigskip

%%\paragraph*{\textbf{Does the removal of the Master Authority affect the security of our scheme?}} 
%\noindent \textbf{Does the removal of the Master Authority affect the security of our scheme?} 
%While the existence of a fully trusted authority can be seen as a mean to improve the security of our construction, this is not true in our case. In particular, the sole reason for assuming the existence of MA, is to allow the curator to be ofline. In other words, by removing MA, the functional keys would be constructed by the curator upon request of the analyst. Hence, while removing the assumption of a fully trusted authority does not affect the security of our schemes, it will result to a scenario in which the curator needs to be constantly online.  

\section{Conclusions}
\label{sec:conclusions}
We strongly believe that in the future cloud-based services will rely less on traditional decryption of information and more on computations over encrypted data. With this in mind, we proposed \LotS; a hybrid protocol based on Functional Encryption and differential privacy. Our protocol allows an analyst to periodically query a CSP for the release of statistics, without breaching the individuals' privacy. We hope that this work will inspire researchers and open new questions in the fascinating field of privacy-preserving computations in untrusted clouds, thus allowing us to create a bridge between  the theoretical concepts of FE and real life applications.

\section*{Acknowledgements}
%\label{Acknow}
The research leading to these results has received support from the Innovative Medicines Initiative Joint Undertaking under grant agreement n° 101034366, resources of which are composed of financial contributions from the European Union's Framework Programme Horizon 2020 and EFPIA companies’ in kind contribution.

\bibliographystyle{ieeetr}
\balance
\bibliography{CBMS}
\vspace{12pt}

%\bibliographystyle{ieeetr}
%\balance
%\bibliography{CBMS}
%\vspace{12pt}

\end{document}